\documentclass[pra,prxquantum,twocolumn,english,superscriptaddress,floatfix]{revtex4-2}

\usepackage[english]{babel}
\usepackage{textcomp,xcolor,natbib}
\usepackage{dcolumn}
\usepackage{graphicx}
\graphicspath{{Figures/}}

\usepackage{amsmath, bbm}
\usepackage{latexsym}
\usepackage{amsfonts}   
\usepackage{amssymb}
\usepackage{array}      
\usepackage{epsfig}
\usepackage{txfonts}
\usepackage{xcolor,braket}
\usepackage[colorlinks=true,linkcolor=blue,urlcolor=blue,citecolor=blue]{hyperref}
\usepackage[normalem]{ulem}
\usepackage{soul}
\usepackage{dsfont}

\usepackage{xfrac}  
\usepackage{relsize}

\usepackage{amsthm}
\usepackage{mathrsfs, mathtools, amsmath}

\usepackage[normalem]{ulem}
\usepackage{cancel}
\usepackage{enumitem}

\usepackage{xspace}
\usepackage{orcidlink}

\newcommand{\ketbra}[1]{{\ket{#1}\bra{#1}}}
\newcommand{\id}{{\mathbbm{1}}}
\newcommand{\abs}[1]{\left\lvert #1 \right\rvert}
\newcommand{\norm}[1]{\left\lVert #1 \right\rVert}
\newcommand{\tr}{{\operatorname{tr}}}

\newcommand{\shro}{Schr\"odinger\xspace}
\newcommand{\schro}{\shro}
\newcommand{\heis}{Heisenberg\xspace}

\newtheorem{example}{Example}
\newtheorem{proposition}{Proposition}

\definecolor{lavender}{rgb}{0.75, 0.58, 0.89}

\begin{document}
\title{Schr\"odinger and Heisenberg non-Markovianity in quantum information tasks}

\author{Federico Settimo\, \orcidlink{0000-0002-0123-6950}}
\email{fesett@utu.fi}
\affiliation{Department of Physics and Astronomy,
University of Turku, FI-20014 Turun yliopisto, Finland}

\author{Kimmo Luoma\, \orcidlink{0000-0003-3118-612X}}
\affiliation{Department of Physics and Astronomy,
University of Turku, FI-20014 Turun yliopisto, Finland}

\author{Jyrki Piilo\, \orcidlink{0000-0002-5595-873X}}
\affiliation{Department of Physics and Astronomy,
University of Turku, FI-20014 Turun yliopisto, Finland}

\author{Bassano Vacchini\, \orcidlink{0000-0002-7574-9951}}
\affiliation{Dipartimento di Fisica ``Aldo Pontremoli'', Universit{\`a} degli Studi di Milano, Via Celoria 16, I-20133 Milan, Italy}
\affiliation{Istituto Nazionale di Fisica Nucleare, Sezione di Milano, Via Celoria 16, I-20133 Milan, Italy}

\author{Dariusz Chru\'sci\'nski \orcidlink{0000-0002-6582-6730}}
\affiliation{Institute of Physics, Faculty of Physics, Astronomy and Informatics,
Nicolaus Copernicus University, Grudziadzka 5/7, 87-100 Toru\'{n},
Poland}

\author{Andrea Smirne\, \orcidlink{0000-0003-4698-9304}}
\affiliation{Dipartimento di Fisica ``Aldo Pontremoli'', Universit{\`a} degli Studi di Milano, Via Celoria 16, I-20133 Milan, Italy}
\affiliation{Istituto Nazionale di Fisica Nucleare, Sezione di Milano, Via Celoria 16, I-20133 Milan, Italy}

\begin{abstract}
    Quantum non-Markovianity has been widely studied and connected to the existence of memory effects in the dynamics of open systems.
    Surprisingly, working in the Schr\"odinger or in the Heisenberg picture can provide inequivalent description non-Markovianity: a process can appear to be memoryless in one picture, while displaying memory effects in the other.
    Here, we investigate which kind of memory is relevant for different quantum information tasks.
    Some of them, such as sending information via a noisy channel, require memory in both pictures in order to exhibit revivals in the task performance.
    For others, only one type of memory is sufficient.
    We also provide necessary conditions for non-Markovianity in both pictures by only considering the dynamics in one picture, showing for instance
    that the previously considered witness of Schr\"odinger non-Markovianity in terms of the volume of accessible states does indeed witness non-Markovianity in both pictures at the same time.
\end{abstract}

\maketitle

\section{Introduction}
\label{sec:intro}

Realistic quantum systems inevitably interact with external degrees of freedom, and are therefore to be described as open systems \cite{Breuer-Petruccione, Vacchini-OQS}.
Due to this interaction, information initially encoded in the open system is generally lost to the environment \cite{Buscemi2016}.
However, in generic dynamical regimes, part of this information can flow back into the open system at later times.
When this occurs, the dynamics is said to be non-Markovian and can be associated with the presence of memory effects in the open-system evolution \cite{BLPV-colloquium}.

Several inequivalent definitions of quantum non-Markovianity have been introduced in the literature.
Some of them are based on the non-monotonic behavior of suitable distinguishability quantifiers between states \cite{BLP, BLP-PRA, Wissmann2015, Megier2021, Smirne2022, Settimo-JSD} or of entanglement \cite{RHP}.
Other approaches characterize non-Markovianity in terms of operational versus non-operational criteria \cite{Budini2022}, causal and non-causal revivals \cite{Buscemi2025}, or within the process tensor formalism \cite{Pollock2018, Pollock2018a, Giarmatzi2021, Santos2025}.
Here, we focus on non-Markovianity defined through divisibility properties of the dynamical map \cite{BLP, BLP-PRA, RHP, Chruscinski2011, Chruscinski2022}.

Surprisingly, describing the dynamics of the open system in the \schro or in the \heis picture does not generally lead to the same notion of non-Markovianity \cite{Settimo-SchroHeis}.
For some dynamics, memory effects can be detected only in one picture, while the evolution remains Markovian in the other.
In other cases, instead, both pictures reveal the existence of memory effects.
In each picture, non-Markovianity can be connected to revivals in time of different physical properties or operational tasks.
For instance, \schro non-Markovianity allows for non-monotonicity of entanglement \cite{RHP} or of the guessing probability between states \cite{BLP, BLP-PRA}.
\heis non-Markovianity, on the other hand, allows for revivals of incompatibility \cite{Settimo-SchroHeis}.

In this work, we investigate several additional quantum-information tasks and analyze which type of memory is relevant for each of them.
We show that some tasks require memory effects only in one picture, whereas others require non-Markovianity in both pictures.
In particular, revivals in time of the channel capacity require non-Markovianity both in the \schro and in the \heis picture.
Channel distinguishability, instead, depends on the way the noisy dynamics is used: postprocessing with the noisy channel is sensitive to \schro non-Markovianity \cite{Bae2016}, while preprocessing is sensitive to Heisenberg non-Markovianity.
In a companion paper \cite{SchroHeisBell}, we study the role of non-Markovianity for revivals in time of Bell nonlocality and device-independent quantum key distribution.

The rest of the paper is organized as follows.
In Sec.~\ref{sec:divisibility}, we introduce the concept of divisibility, both in the \schro and in the \heis picture.
In Sec.~\ref{sec:divisibility_both}, we derive general conditions under which revivals in operational tasks require non-Markovianity in both pictures, we characterize classes of dynamics for which non-Markovianity in one picture necessarily implies non-Markovianity in the other, introduce a channel-norm witness capable of detecting non-Markovianity in both pictures, and show that revivals of the volume of accessible states provide another witness of this kind.
In Sec.~\ref{sec:tasks}, we connect these results to revivals in time in the performance of different quantum information tasks or protocols.
In Sec.~\ref{sec:classical}, we extend the discussion to classical stochastic dynamics, establishing a rigorous analogy between \schro and Heisenberg divisibility in the quantum and classical settings.
Lastly, in Sec.~\ref{sec:conclusion}, we provide the conclusions of our work.

\section{Schr\"odinger versus Heisenberg divisibility}
\label{sec:divisibility}
Under the assumption that the system and the environment are initially uncorrelated, with a fixed initial environmental state, the dynamics of the open system alone is described in the \schro picture by a completely positive and trace preserving (CPTP) map $\Phi_t$ such that $\rho(t) = \Phi_t[\rho(0)]$ \cite{Breuer-Petruccione, Vacchini-OQS}.
Alternatively, one can use the \heis picture, in which the states are kept fixed and the operators evolve as $X(t) = \Phi^*_t[X(0)]$, where $\Phi^*_t$ is the dual of $\Phi_t$, defined according to
\begin{equation}
    \tr\big[\Phi_t[\rho]\,X\big] = \tr\big[\rho\,\Phi^*_t[X]\big]\qquad\forall\rho,X,
\end{equation}
and it is a completely positive and unital (CPU) map.

The dynamical map $\Phi_t$ is the solution of the master equation $d\rho/dt = \mathcal L_t[\rho]$, with $\mathcal L_t = \dot\Phi_t\circ\Phi_t^{-1}$ (in the following we assume that $\Phi_t$ is differentiable and invertible) that can be written in the form \cite{Gorini1976, Lindblad1976}
\begin{equation}
    \label{eq:ME_Schro}
    \mathcal L_t[\rho] = -i[H,\rho] + \sum_\alpha\gamma_\alpha\left[L_\alpha\rho L_\alpha^\dagger-\frac12\left\{L_\alpha^\dagger L_\alpha, \rho\right\}\right],
\end{equation}
with the coefficients $\gamma_\alpha$, Lindblad operators $L_\alpha$ and Hamiltonian $H=H^\dagger$ that can depend on time,
and $\gamma_\alpha$ can take on negative values.
Alternatively, one can derive the master equation for the \heis picture dynamics $dX/dt = \mathcal R_t^*[X]$, with $\mathcal R_t^* = \dot\Phi^*_t\circ(\Phi^*_t)^{-1}$, reading
\begin{equation}
    \label{eq:ME_Heis}
    \mathcal R_t^*[\rho] = -i[K,X] + \sum_\alpha\xi_\alpha\left[R_\alpha^\dagger X R_\alpha-\frac12\left\{R_\alpha^\dagger R_\alpha, X\right\}\right],
\end{equation}
where the (possibly negative) coefficients $\xi_\alpha$ and the operators $R_\alpha$ and $K=K^\dagger$ can indeed depend on time.
Notice that in general $\mathcal R_t\ne\mathcal L_t$, unless $[\Phi_s, \Phi_t]=0$ for all times $s,t$ \cite{Settimo-SchroHeis}.

\subsection{Definitions}
\label{subsec:divisibility_def}
The invertibility of the map $\Phi_t$ further allows one to define the two-parameter propagators that describe the evolution forward in time from time $s$ to a later time $t\geq s$.
In the \schro picture, the propagator is such that \cite{Chruscinski2022}
\begin{equation}
    \label{eq:prop_Schro}
    \Phi_t = \Phi^S_{t,s}\circ\Phi_s,\qquad \Phi_{t,s}^S = \Phi_t\circ\Phi_s^{-1},
\end{equation}
so that $\rho(t) = \Phi_{t,s}^S[\rho(s)]$.
In the \heis picture the propagator reads \cite{Settimo-SchroHeis}
\begin{equation}
    \label{eq:prop_Heis}
    \Phi_t^* = {\Phi^H_{t,s}}^*\circ\Phi_s^*,\qquad {\Phi_{t,s}^H}^* = \Phi_t^*\circ(\Phi_s^*)^{-1}
\end{equation}
and it is such that $X(t) = {\Phi_{t,s}^H}^*[X(s)]$.
Notice that it is possible to decompose the \schro picture dynamical map $\Phi_t$ in two ways
\begin{equation}
    \label{eq:left_right_prop}
    \Phi_t = \Phi_{t,s}^S\circ\Phi_s = \Phi_s\circ\Phi_{t,s}^H,
\end{equation}
but only the first propagator $\Phi_{t,s}^S$ describes the intermediate evolution for states, while only the second propagator $\Phi_{t,s}^H$ describes the intermediate evolution for operators in the Heisenberg picture.
{From the role in Eq.~\eqref{eq:left_right_prop}, the propagators $\Phi^S_{t,s}$ and $\Phi^H_{t,s}$ are also referred to as left and right propagator, respectively.
These two possible factorizations correspond naturally to the two pictures. This correspondence motivates our terminology for the left and right propagator.}

Unlike the dynamical map $\Phi_t$, the propagators do not need to be CP, and violations of CP have been connected to memory effects in the dynamics \cite{BLP, BLP-PRA, RHP, BLPV-colloquium, rivas-quantum-nm}.
In the following, we call the dynamics {\it \schro CP divisible} (SCPD) if $\Phi_{t,s}^S$ is CP or {\it \schro non CP divisible} (SnCPD) otherwise.
Similarly, we call it {\it \heis CP divisible} (HCPD) if $\Phi_{t,s}^H$ is CP or {\it \heis non CP divisible} (HnCPD) otherwise.
If both $\Phi_{t,s}^S$ and $\Phi_{t,s}^H$ are CP, we call the dynamics SHCPD, while if none of the two is CP we call it SHnCPD.
Analogous definitions and acronyms without C will refer to the corresponding divisibility properties where positivity is considered instead of complete positivity.

Violations of P divisibility in both pictures can be linked to revivals in time of suitable norms.
In particular, SPD is equivalent to \cite{Kossakowski-necessary, Wimann2012}
\begin{equation}
    \label{eq:contractivity_TD}
    \norm{\Phi_2[X]}_1\le\norm{\Phi_1[X]}_1\qquad\forall X,
\end{equation}
where $\norm X_1 = \tr\abs X/2$ is the trace norm.
This, in turn, is equivalent to a monotonic decrease in time of the guessing probability between states.
Suppose that one of two states $\rho$ or $\sigma$ is prepared, then the probability of correctly guessing which state was chosen is given by \cite{Fuchs-cryptographic, heinosaari-ziman}
\begin{equation}
    \label{eq:guess_states}
    P_{\mathrm{guess}}^{\mathrm{s}}(\rho,\sigma)  = \frac12\big(1+D_1(\rho,\sigma)\big),
\end{equation}
where $D_1$ is the trace distance (TD)
\begin{equation}
    \label{eq:TD}
    D_1(\rho,\sigma) = \frac12\norm{\rho-\sigma}_1 = \max_{E\in\mathcal E(\mathscr H)}\tr\left[E\,(\rho-\sigma)\right]
\end{equation}
and $\mathcal E(\mathscr H)$ is the set of effects $0\le E\le\id$ on the Hilbert space $\mathscr H$.
{If the dynamics is SPD, then the guessing probability between any two states will diminish monotonically in time, while revivals of such probability can be interpreted as memory effects and are possible only for SnPD dynamics.}

Furthermore, if the dynamics is SCPD and acts locally, it cannot increase the entanglement of any state \cite{Bennett1996, Horodecki2009}
\begin{equation}
    \label{eq:monotonicity_ent}
    E\big((\Phi_t\otimes\operatorname{id})[\rho]\big) \le E\big((\Phi_s\otimes\operatorname{id})[\rho]\big) 
\end{equation}
for any times $t\ge s$ and any state $\rho$; here, $E$ is any entanglement monotone \cite{Horodecki2009}.
Therefore, revivals in time of  entanglement have been connected to memory effects being present in the dynamics \cite{RHP}.

HPD, instead, is equivalent to \cite{Paulsen2003, Settimo-SchroHeis}
\begin{equation}
    \label{eq:contractivity_OD}
    \norm{\Phi_2^*[X]}_\infty\le\norm{\Phi_1^*[X]}_\infty\qquad \forall X,
\end{equation}
where $\norm X_\infty = \max_\phi\abs{\braket{\phi\vert X\vert\phi}}$ is the operator norm.
Analogously, it is equivalent to a monotonic behavior of the guessing probability of which of two effects is prepared
\begin{equation}
    \label{eq:guess_effects}
    P_{\mathrm{guess}}^{\mathrm{e}}(E,F) = \frac12\big(1+D_\infty(E,F)\big),
\end{equation}
where $D_\infty$ is the operator distance (OD)
\begin{equation}
    \label{eq:OD}
    D_\infty(E,F) = \norm{E-F}_\infty= \max_{\rho\in\mathcal S(\mathscr H)}\tr\left[(E-F)\,\rho\right]
\end{equation}
and $\mathcal S(\mathscr H)$ is the set of states on the Hilbert space $\mathscr H$.
{Similarly to the \schro picture, HPD implies monotonic decrease in the guessing probability between any two measurements, while HnPD allows for revivals.
This fact allows to interpret HnPD as memory effects in a similar manner as SnPD \cite{Settimo-SchroHeis}.}

Furthermore, HCPD implies monotonicity of incompatibility
\begin{equation}
    \label{eq:monotonicity_incop}
    I\big(\Phi_t^*[E],\Phi_t^*[F]\big) \le I\big(\Phi_s^*[E],\Phi_s^*[F]\big)
\end{equation}
for any $t\ge s$ and pairs of effects $E,F\in\mathcal E(\mathscr H)$.
Here, $I$ is an arbitrary incompatibility monotone \cite{Heinosaari2015-ibc, Heinosaari2015-im}, i.e. any functional such that $I(E,F)=0$ if and only if $E$ and $F$ are compatible and contractive under CPU maps.

Additionally, the divisibility properties of the dynamics can be inferred directly from the master equations \eqref{eq:ME_Schro} and \eqref{eq:ME_Heis}:
SCPD is equivalent to positivity of all coefficients $\gamma_\alpha$, while SPD corresponds to the weaker condition \cite{Kossakowski-necessary}
\begin{equation}
    \sum_\alpha\gamma_\alpha\,\abs{\braket{\phi_i\vert L_\alpha\vert\phi_j}}^2\ge0
\end{equation}
for any orthonormal basis $\{\phi_i\}_i$ and $i\ne j$.
Analogous conditions hold also for H(C)PD, upon exchanging the coefficients $\gamma_\alpha$ with $\xi_\alpha$ and the Lindblad operators $L_\alpha$ with $R_\alpha$.

\subsection{Cases}
\label{subsec:divisibility_cases}
For a generic dynamics, divisibility can hold only in one picture, in both pictures or in none.
Semigroups of CP dynamical maps are indeed SHCPD processes, since in this case $\mathcal L$ does not depend on time,
so that $\mathcal R=\mathcal L$ and hence $\Phi_{t,s}^S = \Phi_{t,s}^H = e^{\mathcal L(t-s)}$, which is CP.


There exist also dynamics that are non divisible either only in one picture or in both \cite{Carollo2026}.
In particular, Ref. \cite{Settimo-SchroHeis} provides an explicit example of dynamics that are not CP-divisible in one picture, while the propagator in the other picture describes unitary dynamics.
In App.~\ref{app:nm_unitary}, we extend this result to arbitrary dynamics that are not isotropic contractions.

On the other hand, there exist classes of dynamics for which the lack of complete positivity of $\Phi^S_{t,s}$ necessarily implies that of $\Phi^H_{t,s}$, independently of the choice of compatible dynamics $\Phi_s$ up to time s; the converse implication from non-CP $\Phi^H_{t,s}$ to non-CP $\Phi^S_{t,s}$ also holds.
A proof of this result is provided in App.~\ref{app:nM_both}, where uniform expansions of the Bloch sphere are considered as intermediate dynamics.


\section{Schr\"odinger and Heisenberg divisibility}
\label{sec:divisibility_both}
In the previous section, we have recalled some {quantities}, such as entanglement or incompatibility, whose evolution is monotonic under either SCPD or HCPD dynamics, independently of what happens in the other picture.
In this section, we provide a general characterization of quantities that, instead, require both SnCPD {\it and} HnCPD in order to have revivals.
After that, we provide necessary conditions and witnesses for the detection of SHnCPD.

\subsection{Necessity of SHCPD}
\label{subsec:necessity_both}
The conditions for monotonicity of entanglement and incompatibility in Eqs.~\eqref{eq:monotonicity_ent} and \eqref{eq:monotonicity_incop}
are violated -- thus witnessing, respectively, SnCPD and HnCPD -- whenever there exists one single state or pair of effects leading to a revival.
Tasks that require SHnCPD will instead be associated with properties that hold for all states and effects.
We formalize this in the two following Propositions, which generalize the results of \cite{SchroHeisBell}.

The first proposition concerns the impact of the dynamics on properties of the quantum states.
\begin{proposition}
    \label{prop:cond_full_nM}
    Consider an evolution $\Phi_t=\Phi_{t,s}^S\Phi_s$.
    Suppose that there exists a subset $C\subseteq\mathcal S(\mathscr H\otimes \mathbb C^k)$ such that
    \begin{enumerate}
        \item $(\Phi_s\otimes\operatorname{id}_k)[\rho]\in C$ for all states $\rho\in\mathcal S(\mathscr H\otimes \mathbb C^k)$,
        \item any local CPTP map $\Lambda$ is such that $(\Lambda\otimes\operatorname{id}_k)[C]\subseteq C$.
    \end{enumerate}
    If there exists a state $\bar{\rho}\in\mathcal S(\mathscr H\otimes \mathbb C^k)$ such that $(\Phi_t\otimes\operatorname{id}_k)[\bar{\rho}]\not\in C$, then the dynamics must be SHnCPD.

    The same holds if a subset $\bar C\subseteq \mathcal B(\mathscr H\otimes\mathbb C^k)$ is preserved by CPU maps and contains all operators at time $s$.
\end{proposition}
\begin{proof}
    Let $k=1$, but the proof is analogous for $k> 1$.
    If $\Phi_{t,s}^S$ is CPTP then
    \begin{equation}
        \Phi_t[\rho] = \Phi_{t,s}^S\big[\Phi_s[\rho]\big]\in C\qquad\forall\rho
    \end{equation}
    because of 2. on $\Phi_{t,s}^S$.
    If, instead, $\Phi_{t,s}^H$ is CPTP, then
    \begin{equation}
        \Phi_t[\rho] = \Phi_s\big[\Phi_{t,s}^H[\rho]\big]=\Phi_s[\rho^\prime]\in C\qquad\forall\rho
    \end{equation}
    since $\rho^\prime=\Phi_{t,s}^H[\rho]\in\mathcal S(\mathscr H)$ and $\Phi_s$ maps all states to $C$ because of 1.
    Thus, the existence of a state $\bar{\rho}$ such that $\Phi_t[\bar\rho]\not\in C$ implies that $\Phi^S_{t_s}$ or $\Phi^H_{t_s}$ or both are not CPTP.
    The dual is analogous.
\end{proof}
The second proposition, instead, regards the dynamics of monotone functionals under CPTP or CPU maps, when maximizing such functionals over all ensembles of states or over all measurement settings.
\begin{proposition}
    \label{prop:cond_full_nM_2}
    Let $M$ be a functional from the space of ensembles of states to $\mathbb R$ and assume it is monotonic under CPTP maps, $M(\Lambda(\mathcal E))\le M(\mathcal E)$ $\forall\mathcal E$ ensemble, $\forall\Lambda$ CPTP.
    {Then, if}
    \begin{equation}
        \label{eq:contractivity_max_monotone}
        \max_\mathcal E M(\Phi_t(\mathcal E)) > \max_\mathcal E M(\Phi_s(\mathcal E)),
    \end{equation}
    {for $t>s$,} the dynamics must be SHnCPD.

    The same holds for $M$ monotonic under CPU maps and considering effects instead of ensembles.
\end{proposition}
\begin{proof}
    Suppose that the dynamics is SCPD, then 
    \begin{equation}
        M(\Phi_t(\mathcal E))= M(\Phi_{t,s}^S(\Phi_s(\mathcal E))) \le  M(\Phi_s(\mathcal E))
    \end{equation}
    by contractivity of $M$ and therefore Eq.~\eqref{eq:contractivity_max_monotone} cannot hold.

    If the dynamics is HCPD, then
    \begin{equation}
        \begin{split}
            \max_\mathcal E M(\Phi_t(\mathcal E)) &= \max_\mathcal E M(\Phi_s(\Phi_{t,s}^H(\mathcal E)))\\
            &= \max_{\mathcal E^\prime = \Phi_{t,s}^H(\mathcal E)}M(\Phi_s(\mathcal E^\prime))\\            
            &\le \max_\mathcal E M(\Phi_s(\mathcal E))
        \end{split}
    \end{equation}
    where the inequality holds because $\mathcal E^\prime = \Phi_{t,s}^H(\mathcal E)$ is also an ensemble and thus the space on which maximization is performed is reduced.

    Thus, Eq.~(\ref{eq:contractivity_max_monotone}) can hold only if the dynamics is neither SCPD, nor HCPD, i.e. if it is SHnCPD.
    The proof for functionals over effects goes along the same lines.  
\end{proof}

We now provide two examples to clarify the relevance of the two Propositions.
\begin{example}[Entanglement breaking channel]
    \label{ex:EBC}
    A channel $\Phi$ is said to be entanglement breaking (EB) if $(\Phi\otimes\operatorname{id})[\rho]$ is not entangled for every bipartite state $\rho$ \cite{Horodecki2003}.
    Suppose that $\Phi_s$ is EB, then the set $C$ of Proposition \ref{prop:cond_full_nM} is the set of separable states, which is preserved by local CPTP maps.
    Then, for $(\Phi_t\otimes\operatorname{id})[\rho]$ to be entangled, the dynamics must be SHnCPD.
\end{example}

\begin{example}[Incompatibility breaking channel]
    \label{ex:IBC}
    A channel $\Phi^*$ is said to be incompatibility breaking (IB) if \cite{Heinosaari2015-ibc}
    \begin{equation}
        I(\Phi^*[E], \Phi^*[F]) = 0\qquad\forall E,F.
    \end{equation}
    Then, the monotone $M$ of Proposition \ref{prop:cond_full_nM_2} is the incompatibility monotone $I$, which is contractive under CPU maps.
    Then, if $\Phi_s^*$ is IB and $\Phi_t^*$ is not, the dynamics must be SHnCPD.
\end{example}

\subsection{General criteria for non-Markovianity in both pictures}
\label{subsec:criteria_SHnCPD}
In Sec.~\ref{subsec:divisibility_def}, we recalled how the breakdown of P-divisibility, and hence non-Markovianity, in either the \schro or the Heisenberg picture can be witnessed by the non-monotonicity of, respectively, the trace norm for at least a pair of states or the operator norm for at least a pair of effects.
Here we show that stronger criteria relying on the non-monotonicity of the trace and operator norm over all states and effects allow one to infer, respectively, HnPD for SnPD dynamics and SnPD for HnPD dynamics.
\begin{proposition}
    \label{prop:criteria_SHnCPD}
    Suppose that the dynamics $\Phi_t = \Phi_{t,s}^S\circ\Phi_s$ is SnPD and one of the following conditions holds:
    \begin{enumerate}
        \item The maximal distance $\norm\cdot_1$ between states at time $t$ is greater than at time $s$
        \begin{equation}
            \max_{\rho,\sigma} \norm{\Phi_t\big[\rho-\sigma\big]}_1 > \max_{\rho,\sigma} \norm{\Phi_s\big[\rho-\sigma\big]}_1.
        \end{equation}
        \item For all states $\rho\ne\sigma$, it holds that
        \begin{equation}
            \label{eq:rev_TD_all_states}
            \norm{\Phi_t\big[\rho-\sigma\big]}_1 >  \norm{\Phi_s\big[\rho-\sigma\big]}_1.
        \end{equation}
    \end{enumerate}
    Then the dynamics is SHnPD.

    Conversely, if the dynamics is HnPD and either
    \begin{equation}
        \max_{E,F} \norm{\Phi_t^*\big[E-F\big]}_\infty > \max_{E,F} \norm{\Phi_s^*\big[E-F\big]}_\infty,
    \end{equation}
    or
    \begin{equation}\label{eq:eff}
        \norm{\Phi_t^*\big[E-F\big]}_\infty >  \norm{\Phi_s^*\big[E-F\big]}_\infty,
    \end{equation}
    for all effects $E\ne F$, then the dynamics is SHnPD.
\end{proposition}
\begin{proof}
    Suppose that the dynamics is HPD and let $\Delta = \rho-\sigma\ne0$ then
    \begin{equation}
        \label{eq:contractivity_max_TD}
        \begin{split}
            \max_\Delta\norm{\Phi_t\Delta}_1
        &= \max_\Delta\norm{\Phi_s\left[\Phi_{t,s}^H[\Delta]\right]}_1\\
        &= \max_{\Delta^\prime = \Phi_{t,s}^H[\Delta]}\norm{\Phi_s\Delta^\prime}_1\\
        &\le \max_\Delta\norm{\Phi_s\Delta}_1,
        \end{split}
    \end{equation}
    where the inequality holds since {the maximum on $\Delta$ is over Hermitian traceless operators with $\norm\Delta_1\le1$, while $\Delta^\prime$ has norm smaller than $\Delta$, due to contractivity of $\norm\cdot_1$ under CPTP maps.}
    Let now $\bar\rho$ and $\bar\sigma$ be the states corresponding to $\Delta$ maximizing the trace norm at time $s$, then
    \begin{equation}
        \begin{split}
            \norm{\Phi_t\big[\bar\rho-\bar\sigma\big]}_1 &\le\max_{\rho,\sigma}\norm{\Phi_t\big[\rho-\sigma\big]}_1 \\
            &\le\max_{\rho,\sigma}\norm{\Phi_s\big[\rho-\sigma\big]}_1\\
            &=\norm{\Phi_s\big[\bar\rho-\bar\sigma\big]}_1,
        \end{split}
    \end{equation}
    where in the second inequality Eq.~\eqref{eq:contractivity_max_TD} was used.
    Therefore, HPD implies the existence of two states whose distance $\norm\cdot_1$ decreases from $s$ to $t$.
    By taking the negation of this implication, we obtain that the second condition also implies HnPD.

    The proof for the first part of the Proposition goes along the same lines.
\end{proof}






Notice in particular that SnPD implies that Eq.~\eqref{eq:rev_TD_all_states} holds for at least one pair of states;
if it holds for all states, then the dynamics must be also HnPD. Analogously, HnPD implies Eq.~(\ref{eq:eff}) for at least one pair
of effects, and if it holds for all effects it also implies SnPD.

\subsection{Witness for non-Markovianity in both pictures}
\label{subsec:witness_SHnCPD}
We now provide a witness for non-Markovianity in both pictures, i.e. for SHnCPD, that combines the witness for SnPD and HnPD of Eqs.~\eqref{eq:contractivity_TD} and Eq.~\eqref{eq:contractivity_OD} respectively.
Such witness relies on the fact that SnPD and HnPD imply contractivity of the guessing probability between states or effects of Eqs.~\eqref{eq:guess_states} and \eqref{eq:guess_effects}.
It is possible to combine such two guessing probabilities by defining the {semi}norm of the channel $\Phi$ as
\begin{equation}
    \norm\Phi_{1\to1}^{{(0)}}\coloneqq \max_{\Delta\in\mathcal S_1}\max_{E\in\mathcal S_\infty}\tr\big[\Phi[\Delta]\,E\big],
\end{equation}
where {$\mathcal S_1\coloneqq\big\{X\in\mathcal B(\mathscr H)\,\vert\,X = X^\dagger,\,\norm X_1\le1,\,\tr[X]=0\big\}$ and $\mathcal S_\infty\coloneqq\big\{X\in\mathcal B(\mathscr H)\,\vert\,X = X^\dagger,\,\norm X_\infty\le1\big\}$.}
Notice that if one only maximizes over effects (states), then the TD (OD) is recovered, as in Eqs.~\eqref{eq:TD}, \eqref{eq:OD}.

\begin{proposition}
    \label{prop:revivals_norm_SHnCPD}
    If the {semi}norm $\norm\cdot_{1\to1}^{{(0)}}$ is non-monotonic in time
    \begin{equation}
        \norm{\Phi_t}_{1\to1}^{{(0)}} >\norm{\Phi_s}_{1\to1}^{{(0)}},\qquad t>s,
    \end{equation}
    then the dynamics SHnCPD.
\end{proposition}
\begin{proof}
    Notice that any CPTP map sends $\mathcal S_1$ into itself, while any CPU map sends $\mathcal S_\infty$ into itself.
    First suppose that the dynamics is SCPD, so that
    \begin{equation}
        \begin{split}
            \norm{\Phi_t}_{1\to1}^{{(0)}} &= \norm{\Phi_{t,s}^S\circ\Phi_s}_{1\to1}^{{(0)}} = \max_{\Delta\in\mathcal S_1}\max_{E\in\mathcal S_\infty}\tr\big[\Phi_{t,s}[\Phi_s[\Delta]]\,E\big]\\
            &= \max_{\Delta\in\mathcal S_1}\max_{E\in\mathcal S_\infty}\tr\big[\Phi_s[\Delta]\,{\Phi_{t,s}^S}^*[E]\big]\\
            &= \max_{\Delta\in\mathcal S_1}\max_{E^\prime\in{\Phi_{t,s}^S}^*[\mathcal S_\infty]}\tr\big[\Phi_s[\Delta]\,E^\prime\big]\\
            &\le \max_{\Delta\in\mathcal S_1}\max_{E\in\mathcal S_\infty}\tr\big[\Phi_s[\Delta]\,E\big] = \norm{\Phi_s}_{1\to1}^{{(0)}}
        \end{split}
    \end{equation}
    since ${\Phi_{t,s}^S}^*[\mathcal S_\infty]\subseteq\mathcal S_\infty$.
    If, instead, it is HCPD, contractivity follows similarly
    \begin{equation}
        \begin{split}
            \norm{\Phi_t}_{1\to1}^{{(0)}} &= \norm{\Phi_s\circ\Phi_{t,s}^H}_{1\to1}^{{(0)}} = \max_{\Delta\in\mathcal S_1}\max_{E\in\mathcal S_\infty}\tr\big[\Phi_{t,s}^H[\Delta]\,\Phi_s^*[E]\big]\\
            &= \max_{\Delta\in\Phi_{t,s}^H[\mathcal S_1]}\max_{E\in\mathcal S_\infty}\tr\big[\Delta\,\Phi_s^*[E]\big]\le\norm{\Phi_s}_{1\to1}^{{(0)}},
        \end{split}
    \end{equation}
    where it was used instead $\Phi_{t,s}^H[\mathcal S_1]\subseteq\mathcal S_1$.
\end{proof}

%


\subsection{Volume of accessible measurements}
\label{subsec:volume_measurements}
A widely used witness for SnCPD is the so-called volume of accessible states \cite{Lorenzo2013}.
Here, we show that such witness actually detects SHnCPD.

In \cite{Lorenzo2013}, the time evolution of an infinitesimal volume of the set of quantum states $dV_{\text s}(0)$ was considered.
Due to the interaction with the environment, the volume of accessible states at a later time $t$ evolves as
\begin{equation}
    dV_{\text s}(t) = \abs{\det\Phi_t}\,dV_{\text s}(0) \le dV_{\text s}(0),
\end{equation}
and it decreases monotonically under CPTP maps \cite{Wolf2008}.
Therefore, revivals in time of the volumes of accessible states $dV_{\text s}(t) > dV_{\text s}(s)$ implies that $\Phi_{t,s}^S$ is not CP and were therefore used as a witness for SnCPD.
Such revival is equivalent to $d\abs{\det\Phi_t}/dt >0$.

In an analogous way, by working in the \heis picture, one can also define the volume of accessible measurements as
\begin{equation}
    dV_{\text e}(t) = \abs{\det\Phi_t^*}\,dV_{\text e}(0) \le dV_{\text e}(0).
\end{equation}
In a similar way as done in \cite{Wolf2008}, it can be proven that $dV_{\text e}(t)$ is monotonically decreasing under HCPD dynamics, and therefore revivals witness HnCPD.

On the other hand, it holds that $\det\Phi_t^* = \overline{\det\Phi_t}$, where $\overline z$ represents the complex conjugation, and therefore $\abs{\det\Phi_t^*} = \abs{\det\Phi_t}$.
This, in turns, implies that
\begin{equation}
    \frac{dV_{\text s}(t)}{dV_{\text s}(0)} = \frac{dV_{\text e}(t)}{dV_{\text e}(0)} = \abs{\det\Phi_t}
\end{equation}
and thus revivals of $dV_{\text s}(t)$ are present if and only if there are revivals of $dV_{\text e}(t)$, which implies that the dynamics must be SnCPD and HnCPD.
Therefore, the witness of non-Markovianity in terms of revivals of the volume of accessible states introduced in \cite{Lorenzo2013} is actually a witness of SHnCPD.

\subsection{Choi state witness for non-Markovianity in the Heisenberg picture}
In the \schro picture, there are necessary and sufficient conditions for witnessing both P and CP divisibility.
In particular, SnPD is equivalent to non-monotonicity in $\norm\cdot_1$ \cite{BLP, BLP-PRA}, while SnCPD is equivalent to \cite{RHP}
\begin{equation}
    \label{eq:RHP_Schro}
    g_S(t) = \lim_{\varepsilon\to0^+}\frac{\norm{\left(\Phi_{t+\varepsilon,t}^S\otimes\operatorname{id}\right)\left[\ketbra\Psi\right]}_1-1}{\varepsilon}>0,
\end{equation}
where $\Psi$ is the maximally entangled state
\begin{equation}
    \ket\Psi = \frac1{\sqrt d}\sum_i\ket i\otimes\ket i.
\end{equation}
HnPD, instead, is equivalent to non-monotonicity of $\norm\cdot_\infty$ \cite{Settimo-SchroHeis}.
We now introduce a necessary and sufficient condition for HnCPD analogous to Eq.~\eqref{eq:RHP_Schro}.

It would be tempting to introduce a witness by using ${\Phi_{t+\varepsilon,t}^H}^*$ instead of $\Phi_{t+\varepsilon,t}^S$ and $\norm\cdot_\infty$ instead of $\norm\cdot_1$ in Eq.~\eqref{eq:RHP_Schro}.
However, this would not give a necessary condition for HnCPD, since there exist PU but not CPU maps such that $\norm{\left(\Lambda^*\otimes\operatorname{id}\right)\ketbra\Psi}_\infty = \norm{\ketbra\Psi}_\infty$, as, for instance, for $\Lambda$ being the partial transposition.

However, it is possible to derive a necessary and sufficient condition for HnCPD as
\begin{equation}
    \label{eq:RHP_Heis}
    g_H(t) = \lim_{\varepsilon\to0^+}\frac{\norm{\left({\Phi_{t+\varepsilon,t}^H}^*\otimes\operatorname{id}\right)\left[\ketbra\Psi\right]}_1-1}{\varepsilon}>0.
\end{equation}
Notice that the norm $\norm\cdot_1$ is used since, for any CPU map, $\Lambda^*\otimes\operatorname{id}$ is also TP on $\operatorname{span}\{\ketbra\Psi\}$:
\begin{equation}
    \begin{split}
        \tr\left[\left(\Lambda^*\otimes \operatorname{id}\right)[\ketbra\Psi]\right] &= \tr_1\left[\tr_2\left[\left(\Lambda^*\otimes\operatorname{id}\right)[\ketbra\Psi\right]\right] \\
        &= \frac1d\tr_1\left[\Lambda^*\left[\id\right]\right] = \frac1d\tr_1\left[{\id}\right] =1.
    \end{split}
\end{equation}
{Therefore, the Choi operator $\left(\Lambda^*\otimes \operatorname{id}\right)[\ketbra\Psi]$ is a density matrix, and thus its norm $\norm\cdot_1$ is contractive under composition with CP maps.}
Using Eq.~\eqref{eq:contractivity_TD} and the results of \cite{RHP}, it is thus possible to conclude that $g_H(t)>0$ if and only if the dynamics is HnCPD.

Furthermore, the quantity in Eq.~\eqref{eq:RHP_Heis} can be written in terms of the coefficients of the \heis picture master equation \eqref{eq:ME_Heis}.
Using the fact that $g_H(t)\ne0$ is equivalent to the positivity of $(\id-\ketbra\Psi)\mathcal R^*_t[\ketbra\Psi](\id-\ketbra\Psi)$, one can in fact show that \cite{Hall2014}
\begin{equation}
    g_H(t) = \frac 2d\sum_\alpha\max\left\{0,\,-\xi_\alpha(t)\right\}.
\end{equation}
Similarly, in the \schro picture one has \cite{Hall2014}
\begin{equation}
    g_S(t) = \frac 2d\sum_\alpha\max\left\{0,\,-\gamma_\alpha(t)\right\}.
\end{equation}

\section{Relevant quantum information tasks}
\label{sec:tasks}
We now apply the results of Sec.~\ref{sec:divisibility_both} to relevant quantum information tasks and provide a characterization of whether SnCPD, HnCPD, or SHnCPD is required in order to have revivals in time of the task performance.
One such task was already presented in \cite{SchroHeisBell}, in which we demonstrated that SHnCPD is required in order to have revivals in times of the violation of Bell inequalities for any state and set of local measurements.

\subsection{Capacity}
\label{subsec:capacity}
As a first relevant quantum information task, we consider that of transmitting information through a quantum channel.
The amount of information that can be transmitted via a quantum channel is quantified by the capacity of the latter and, depending on the task at hand, there are different definitions for it.

If one wants to transmit classical information using a quantum channel, the amount of transmitted information per use of the channel is quantified by its classical capacity \cite{Holevo2012}
\begin{equation}
    C_C(\Phi) = \max_{\{p_i,\rho_i\}} S\left(\sum_i p_i\Phi[\rho_i]\right)-\sum_i p_i S\left(\Phi[\rho_i]\right),
\end{equation}
where $S(\rho) = -\tr[\rho\log\rho]$ is the von Neumann entropy.

If, instead, one wants to transmit classical information via the quantum channel, but allowing for the system to be entangled with some ancilla, the channel capacity is the so-called entanglement-assisted classical capacity \cite{Bennett2002, Wilde2013}
\begin{equation}
    C_E(\Phi) = \max_\rho S\left(\rho\right) + S\left(\Phi[\rho]\right) - S\left((\Phi\otimes\operatorname{id})\left[\ketbra{\psi_\rho}\right]\right),
\end{equation}
where $\psi_\rho$ is any purification of $\rho$.

Lastly, if one wants to transmit quantum information, the amount of transmittable information is given by the quantum capacity \cite{Lloyd1997}
\begin{equation}
    C_Q(\Phi) = \lim_{n\to\infty}\frac1n\max_\rho I_C\left(\rho,\Phi^{\otimes n}\right),
\end{equation}
where $I_C(\rho,\Lambda) = S(\Lambda[\rho]) - S(\Lambda^c[\rho])$ and $\Lambda^c$ is the complementary channel \cite{Holevo2012}.

All such capacities are contractive under 
the action of CP maps, i.e. if $\Phi_{t,s}^S$ is CPTP then
\begin{equation}
    C_{\alpha}\left(\Phi_{t,s}^S\circ\Phi_s\right) \le C_\alpha\left(\Phi_s\right),\qquad \alpha = C,E,Q
\end{equation}
and therefore revivals in any such capacity can be used to witness SnCPD.
In particular, a measure of non-Markovianity in terms of the capacity can be introduced as \cite{Bylicka2014}
\begin{equation}
    \label{eq:non-Markov_capacity}
    \mathcal N_\alpha(\Phi) = \int_{\frac d{dt}C_\alpha(\Phi_t)>0}dt\,\frac d{dt}C_\alpha(\Phi_t)\qquad \alpha = C,E,Q.
\end{equation}
Naturally, it holds that $\mathcal N_\alpha(\Phi)=0$ whenever the dynamics is SCPD.

However, it is easy to verify that the capacity is contractive also under HCPD dynamics, i.e.
\begin{equation}
    \label{eq:contractivity_C_Heis}
    C_\alpha\left(\Phi_s\circ\Phi_{t,s}^H\right)\le C_\alpha\left(\Phi_s\right)
\end{equation}
if $\Phi_{t,s}^H$ is CPTP.
The contractivity of the capacity follows from Prop.~\ref{prop:cond_full_nM_2}, since $C_\alpha$ can always be written as the maximum over all possible input ensembles of a monotonic functional under CPTP maps.
Therefore, the measure of non-Markovianity of Eq.~\eqref{eq:non-Markov_capacity} can be non zero only if both $\Phi_{t,s}^S$ and $\Phi_{t,s}^H$ are not CP, i.e. if the dynamics is SHnCPD.
In other words, in order to have revivals in time of the channel capacity, then one needs both SnCPD and HnCPD dynamics.

\begin{example}[Capacity]
    \label{ex:capacity}
    We now provide an example of a dynamics which is SnCPD but HCPD, for which the capacity does indeed diminish monotonically in time.
    We assume that the dynamics up to time $s$ is of dephasing form
    \begin{equation}
        \label{eq:dephasing}
        \Phi_s[\rho] = \Phi_s^*[\rho] = \frac{1+d}2\,\rho + \frac{1-d}2\,\sigma_z\rho\sigma_z,
    \end{equation}
    with the dephasing parameter $d$ assumed to be a monotonic function of time $s$.
    Then, between $s$ and $t$, the dynamics is described in the \heis picture by the unitary transformation
    \begin{equation}
        \label{eq:unitary_after_dephasing}
        {\Phi_{t,s}^H}^*[X] = U_{t-s}^\dagger X U_{t-s}, \qquad U_{t-s}=e^{-i \sigma_y (t-s)}.
    \end{equation}
    If one fixes the final time $t = s+\pi/4$, then the unitary is such that $U_{\pi/4}\sigma_z U^\dagger_{\pi/4} = \sigma_x$ and $U_{\pi/4}\sigma_x U_{\pi/4}^\dagger = -\sigma_z$.
    The dynamical map $\Phi_t$ at time $t$ is also of dephasing form, but in the $\sigma_x$ direction instead of the $\sigma_z$ direction of $\Phi_s$ and with the same dephasing parameter $d$, up to unitaries.
    Such dynamics is HCPD, but it is SnCPD \cite{SchroHeisBell}.
    
    The classical capacity $C_C(\Phi)$ is trivially constant over time, since one can perfectly transmit classical information both at time $s$ by encoding it in the eigenbasis of $\sigma_z$ (due to the fact that $\Phi_s[\sigma_z] = \sigma_z$) and at time $t$ by encoding it in the eigenbasis of $\sigma_x$ (since $\Phi_t[\sigma_x] = \sigma_x$).
    Therefore
    \begin{equation}
        C_C\left(\Phi_t\right) = C_C\left(\Phi_s\right) = 1.
    \end{equation}
    Notice that the maximal classical capacity is achieved at all times.
    
    For the entanglement-assisted classical capacity, it can be explicitly written as \cite{Bylicka2014, Devetak2005_capacity}
    \begin{equation}
        C_E\left(\Phi_s\right) = 2-h_2\left(\frac{1+d}2\right), 
    \end{equation}
    where $h_2$ is the binary entropy $h_2(p) = -p\log_2p - (1-p)\log_2(1-p)$.
    But, since $\Phi_t$ is also of dephasing form with the same dephasing parameter, it holds that
    \begin{equation}
        C_E\left(\Phi_t\right) = C_E\left(\Phi_s\right) = 2-h_2\left(\frac{1+d}2\right),
    \end{equation}
    and therefore also $C_E$ does not present any revival.
    Notice that the entanglement-assisted classical capacity is strictly smaller than the maximal value $C_E(\operatorname{id})=2$ unless $d=1$.
    Lastly, the quantum capacity can be rewritten as $C_Q(\Phi_t) = C_E(\Phi_t) - 1$, and therefore does not present revivals either.
    Therefore, for the dynamics of Eqs.~\eqref{eq:dephasing}, \eqref{eq:unitary_after_dephasing}, which is SnCPD but HCPD, it holds that $\mathcal N_C(\Phi) = \mathcal N_E(\Phi)=\mathcal N_Q(\Phi)=0$, as expected from Proposition \ref{prop:cond_full_nM_2}.
\end{example}

{Notice that in this example the dynamics is SnCPD, while the \heis picture propagator is just a unitary, so that the evolution from $s$ to $t$ can be described in the Heisenberg picture without introducing any additional environmental degrees of 
freedom; in the Schrödinger picture, in contrast, the corresponding intermediate propagator is not CP and therefore cannot itself be interpreted as a physical quantum channel acting on the system. 
The opposite situation, in which the Schrödinger propagator is CP and unitary while the Heisenberg one is not, occurs in Example \ref{ex:distinguishability}. Thus, depending on the picture used to describe the dynamics, the same open-system evolution can admit qualitatively different intermediate descriptions, which may be relevant when choosing a representation for its simulation or analysis.}

\subsection{Channel distinguishability}
\label{subsec:channel_distinguishability}
Let us now consider the scenario where one wants to distinguish between quantum channels.
Let $\Lambda_1$ and $\Lambda_2$ be two quantum channels prepared respectively with probabilities $p$ and $1-p$, then their distinguishability is given by \cite{Aharonov1998, Watrous2018}
\begin{equation}
    \label{eq:distinguish_channels}
    P_{\mathrm{guess}}^{\mathrm{ch}}\left(\Lambda_1,\Lambda_2,p\right) =\frac12\big(1+ \norm{p\Lambda_1-(1-p)\Lambda_2}_\diamond\big),
\end{equation}
where $\norm\cdot_\diamond$ is the so called diamond norm 
\begin{equation}
    \label{eq:diamond_norm}
    \norm\Lambda_\diamond = \max_{X,\norm X_1\le1}\norm{\left(\Lambda\otimes\operatorname{id}\right)[X]}_1,
\end{equation}
where $X$ is a bipartite operator with an ancilla that has the same dimension of the system, satisfying $\norm X_1\le1$.


The diamond norm is monotonic under postprocessing, i.e. if $\Phi$ is another CP channel, then $\norm{\Phi\circ\Lambda}_\diamond\le\norm{\Lambda}_\diamond$.
Therefore, if $\Phi_t$ is a SCPD dynamics, the guessing probability $P_{\mathrm{guess}}^{\mathrm{ch}}(\Phi_t\circ\Lambda_1,\Phi_t\circ\Lambda_2,p)$ is monotonically decreasing in time.
In fact, SCPD holds if and only if \cite{Bae2016}
\begin{equation}
    \label{eq:contractivity_guess_channels}
    P_{\mathrm{guess}}^{\mathrm{ch}}(\Phi_t\circ\Lambda_1,\Phi_t\circ\Lambda_2,p)\le P_{\mathrm{guess}}^{\mathrm{ch}}(\Phi_s\circ\Lambda_1,\Phi_s\circ\Lambda_2,p)
\end{equation}
for all $p\in[0,1]$, for all $t>s$ and for all CP maps $\Lambda_1$, $\Lambda_2$.
In other words, postprocessing with a SCPD dynamics $\Phi_t$ will diminish the distinguishability between any two quantum channels.

We now investigate what happens if one preprocesses with a CP map $\Phi$ instead of postprocessing, i.e. we consider now the distinguishability between $\Lambda_1\circ\Phi_t$ and $\Lambda_2\circ\Phi_t$.
To do so, we first show that the diamond norm can be written as
\begin{equation}
    \label{eq:norm_diamond_infty}
    \norm{\Lambda^*}_{\rm cb} =    \norm\Lambda_\diamond =  \max_{X,\norm X_\infty\le1}\norm{\left(\Lambda^*\otimes\operatorname{id}\right)[E]}_\infty.
\end{equation}
%
%
%
This implies in a straightforward way that the diamond norm is also contractive under preprocessing $\norm{\Lambda\circ\Phi}_\diamond\le\norm{\Lambda}_\diamond$.
Then, if $\Phi_t$ is HCPD, it holds that for all channels $\Lambda$
\begin{equation}
    \norm{\Lambda\circ\Phi_t}_\diamond=\norm{\Lambda\circ\Phi_s\circ\Phi_{t,s}^H}_\diamond\le\norm{\Lambda\circ\Phi_s}_\diamond.
\end{equation}
Therefore, preprocessing with a HCPD dynamics $\Phi_t$ will diminish the distinguishability between any two quantum channels.
Conversely, the distinguishability is also diminished by postprocessing with a SCPD dynamics.

By writing the diamond norm in terms of the norm $\norm\cdot_\infty$ as in Eq.~\eqref{eq:norm_diamond_infty} and using the fact that $\norm{(\Phi\otimes \operatorname{id})[X]}_\infty\le\norm X_\infty$ if and only if $\Lambda$ is CPU \cite{Paulsen2003}, it follows that
\begin{equation}
    \label{eq:contractivity_guess_channels_preproc}
    P_{\mathrm{guess}}^{\mathrm{ch}}(\Lambda_1\circ\Phi_t,\Lambda_2\circ\Phi_t,p)\le P_{\mathrm{guess}}^{\mathrm{ch}}(\Lambda_1\circ\Phi_s,\Lambda_2\circ\Phi_s,p)
\end{equation}
for all $p\in[0,1]$, $t>s$ and channels $\Lambda_1$, $\Lambda_2$ if and only if the dynamical map $\Phi_t$ is HCPD.
This result is the analog of Eq.~\eqref{eq:contractivity_guess_channels} of \cite{Bae2016} for preprocessing instead of postprocessing, 
and HCPD instead of SCPD.

\begin{example}[Channel distinguishability]
    \label{ex:distinguishability}
    Suppose that we want to distinguish between dephasing and the identity channel
    \begin{equation}
        \Lambda_1[\rho] = \frac34\rho+\frac14\sigma_z\rho\sigma_z,\qquad\Lambda_2=\operatorname{id},
    \end{equation}
    each with the same probability $p=1/2$.
    Assume that we preprocess or postprocess both channels with a noisy channel similar to that of Example \ref{ex:capacity}.
    In particular, we assume $\Phi_s$ to be as in Eq.~\eqref{eq:dephasing}, and
    \begin{equation}
        \label{eq:unitary_after_dephasing_Schro}
        \Phi_{t,s}^S[\rho] = U \rho U^\dagger, \qquad U=e^{-i \pi \sigma_y/4}.
    \end{equation}
    Such dynamics is SCPD but HnCPD \cite{SchroHeisBell}.
    From the results of \cite{Bae2016}, it holds that the channel distinguishability is monotonic under postprocessing and, since $\Phi_{t,s}^S$ is unitary, it holds that
    \begin{equation}
        P_{\mathrm{guess}}^{\mathrm{ch}}(\Phi_t\circ\Lambda_1,\Phi_t\circ\Lambda_2)= P_{\mathrm{guess}}^{\mathrm{ch}}(\Phi_s\circ\Lambda_1,\Phi_s\circ\Lambda_2) = \frac d8,
    \end{equation}
    where we neglected the third argument $p=1/2$ for compactness.

    However, by virtue of HnCPD and Eq.~\eqref{eq:contractivity_guess_channels_preproc}, the distinguishability is not monotonic under preprocessing.
    Indeed, by direct calculation, it holds that
    \begin{equation}
        \begin{split}
            \frac12\left(1+\frac 18\right) &=P_{\mathrm{guess}}^{\mathrm{ch}}(\Lambda_1\circ\Phi_t,\Lambda_2\circ\Phi_t)\\
            &> P_{\mathrm{guess}}^{\mathrm{ch}}(\Lambda_1\circ\Phi_s,\Lambda_2\circ\Phi_s)=\frac12\left(1+\frac d8\right)
        \end{split}
    \end{equation}
    for any $d\ne 1$.
    Therefore, we have shown an explicit channel distinguishability task whose performance is monotonic under postprocessing but not under preprocessing, since the dynamics is SCPD but HnCPD.
\end{example}

\section{Classical stochastic dynamics}
\label{sec:classical}
We now derive the analogous results of the previous section for a classical stochastic dynamics.
Assume that the dynamics of a $d$-sites probability distribution evolves as $\mathbf p(t) = {T}(t)\mathbf p(0)\in\mathbb P_d$, where $T(t)$ is a $d\times d$ stochastic matrix obeying $T_{ij}(t)\ge0$ and $\sum_i T_{ij}(t)=1$, while $\mathbb P_d$ is the set of probability distributions on $d$ sites.
It is possible to write the master equation for $T$ in two ways
\begin{equation}
    \label{eq:ME_classical}
    \dot T(t) = L(t) T(t) = T(t) R(t),
\end{equation}
where $L(t)$ and $R(t)$ are the classical analogous of the generators $\mathcal L_t$ and $\mathcal R_t$ of Eqs.~\eqref{eq:ME_Schro}, \eqref{eq:ME_Heis} and they obey $\sum_i L_{ij}(t) = \sum_iR_{ij}(t) =0$.
The master equation for the probability vector $\mathbf p(t)$ reads
\begin{equation}
    \dot{\mathbf p}(t) = L(t)\mathbf p(t),
\end{equation}
while there is no master equation for $\mathbf p(t)$ in terms of $R(t)$.
Alternatively, one can define a classical version of the \heis picture by introducing a vector of observables $\mathbf x(t) = T^\top(t)\mathbf x(0)\in\mathbb O_d$, where $^\top$ denotes the transposition and $\mathbb O_d = \{\mathbf x\in\mathbb R^d\,\vert\, 0\le x_i\le1\}$, corresponding to the set of classical effects on $d$ sites;
indeed it holds that
\begin{equation}
    \big(\mathbf x(0),\mathbf p(t)\big) = \big(\mathbf x(t),\mathbf p(0)\big),
\end{equation}
with $(\mathbf a,\mathbf b) = \sum_ia_ib_i$.
This \heis picture vector obeys the master equation 
\begin{equation}
    \dot{\mathbf x}(t) = R^\top(t)\mathbf x(t).
\end{equation}

Similarly to Eqs.~\eqref{eq:prop_Schro}, \eqref{eq:prop_Heis}, one can define the classical propagators
\begin{equation}
    T(t) = T_S(t,s) T(s) = T(s) T_H(t,s),
\end{equation}
for which
\begin{equation}
    \mathbf p(t) = T_S(t,s)\mathbf p(s),\qquad \mathbf x(t) = T_H^\top(t,s) \mathbf x(s)
\end{equation}
for all $t>s$.
The dynamics is therefore said to be SPD if $T_S(t,s)$ is a valid stochastic matrix, and HPD if $T_H(t,s)$ is a stochastic matrix.
Notice that in the classical case P and CP are equivalent.
It holds that SPD is equivalent to
\begin{equation}
    \frac d{dt}\norm{T(t)\mathbf p}_1\le0\qquad\forall \mathbf p,
\end{equation}
while HPD is equivalent to \cite{Settimo-SchroHeis}
\begin{equation}
    \frac d{dt}\norm{T^\top(t)\mathbf x}_\infty\le0\qquad\forall \mathbf x.
\end{equation}
Like for the quantum case, SPD and HPD are not equivalent but, unlike the quantum version, at least for $d=2$ HPD implies SPD \cite{Settimo-SchroHeis}.

\subsection{Criteria for classical SHnPD}
\label{subsec:criteria_SHnPD_class}
We notice that all results obtained so far for quantum dynamical maps also hold for classical stochastic processes, upon using stochastic maps instead of CPTP maps.
For instance, the analog of Proposition \ref{prop:cond_full_nM} reads as follows.
\begin{proposition}
    Let $T(t) = T_S(t,s) T(s)$ and suppose that there exists a subset $C\subseteq \mathbb P_d$ such that
    \begin{itemize}
        \item $T(s)\mathbf p\in C$ for all $\mathbf p\in\mathbb P_d$,
        \item any stochastic map $S$ is such that $S(C)\subseteq C$.
    \end{itemize}
    If there exists $\bar{\mathbf p}$ such that $T(t)\bar{\mathbf p}\not\in C$, then the dynamics must be SHnPD.
\end{proposition}
The proof is analogous to the proof of Proposition \ref{prop:cond_full_nM}.
In a similar way, it is possible to derive the classical counterparts also of Propositions \ref{prop:cond_full_nM_2} and \ref{prop:criteria_SHnCPD}.
In particular, Proposition \ref{prop:criteria_SHnCPD} holds due to the fact that also in the classical case it is possible to write the norm $\norm\cdot_1$ as
\begin{equation}
    \norm{\mathbf p}_1 = \max_{\mathbf x}(\mathbf x,\mathbf p),
\end{equation}
where the maximization is taken over vectors $\mathbf x$ such that $\abs{x_i}\le1$ and the proof then follows in a similar way.

\subsection{Classical channel capacity}
\label{subsubsec:classical_capacity}
We now show that revivals in time of classical channel capacity also requires SHnPD, similar to the quantum counterpart of Sec.~\ref{subsec:capacity}.
The capacity of a classical stochastic channel {$T$} reads \cite{Shannon1948, Cover2005}
\begin{equation}
    \label{eq:classical_capacity}
    C(T) = \max_{\mathbf p} I(\mathbf p : T\mathbf p),
\end{equation}
where $I$ is the mutual information $I(X:Y) = H(Y)-H(Y\vert X)$ and $H$ is the Shannon entropy $H(\mathbf p) = -\sum_i p_i\log p_i$.
The capacity is monotonic under postprocessing $C(S\,T)\le C(T)$ for any stochastic channels $S$ and $T$, and therefore SPD implies
\begin{equation}
    C\big(T(t)\big) = C\big(T_S(t,s)\,T(s)\big)\le C\big(T(s)\big).
\end{equation}

In App.~\ref{app:contractivity_cl_capacity}, we show that the capacity is also monotonic under HPD dynamics
\begin{equation}
    C\big(T(t)\big) = C\big(T(s)\,T_H(t,s)\big)\le C\big(T(s)\big)
\end{equation}
whenever $T_H(t,s)$ is a stochastic matrix.
Therefore, like in the quantum case, revivals in time of the classical capacity $C(T(t))$ require that the stochastic dynamics is SHnPD.
This implies that divisibility both in the \schro and in the \heis picture is relevant not only for open quantum system dynamics but also for classical stochastic dynamics.


{\section{Conclusions and outlook}}
\label{sec:conclusion}
In this work, we investigated the role of \schro and Heisenberg non-Markovianity in quantum information tasks.
While the two notions of divisibility are generally inequivalent, we showed that they are associated with different kinds of memory effects and can therefore become relevant for different operational scenarios.
In particular, some tasks, such as revivals of entanglement or incompatibility, require non-Markovianity only in one picture, whereas others, including revivals of channel capacities, require non-Markovianity in both pictures simultaneously.

We further derived general conditions under which revivals in certain quantities necessarily imply non-Markovianity in both pictures, characterized classes of dynamics for which non-Markovianity in one picture implies non-Markovianity in the other, and introduced witnesses capable of detecting non-Markovianity in both pictures.
In particular, we showed that revivals of the volume of accessible states actually witness this stronger form of non-Markovianity.
Lastly, we extended our analysis to classical stochastic processes, showing that the distinction between \schro and Heisenberg divisibility is not exclusive to the quantum setting.

Our results therefore provide a unified perspective on memory effects in both quantum and classical dynamics, and clarify which notion of non-Markovianity is relevant for different information-processing tasks.
{An interesting question is whether analogous considerations apply to other representations of the dynamics, such as the interaction picture.
Since the interaction picture is defined with respect to a chosen reference evolution, the corresponding notion of divisibility and its operational interpretation may as well depend on this choice.}

\section*{Acknowledgements}
FS acknowledges support from Magnus Ehrnroothin S\"a\"ati\"o.
DC was supported by the Polish National Science Center under Projects No. 2024/55/B/ST2/01781.
The authors thank the Toru\'n group and the Aleksander Jab\l o\'nski Foundation for hospitality received.

\appendix
\section{Non Markovianity only in one picture for non-isotropic maps}
\label{app:nm_unitary}
Assume for simplicity that $\Phi_s$ is unital, so that it can be written on the ($d$-dimensional) Bloch sphere $\Phi_s=\operatorname{diag}\{\lambda_1,\ldots,\lambda_d\}$, with $0<\lambda_1\le\ldots\le\lambda_d\le1$ and at least one inequality is not equality, so that the map is not isotropic.
In particular it holds that $\lambda_1<\lambda_d$.
    
Let {$\mathbf r =(0,0,\ldots,1)^\top$} and let $\Phi_{t,s}^S=\mathcal U$ be a unitary transformation such that {$(0,0,\ldots,1)^\top\mapsto(1,\ldots,0,0)^\top$}.
But then
\begin{equation}
    {\Phi_{t,s}^H \mathbf r = \Phi_s^{-1}\mathcal U\Phi_s\mathbf r = \Phi_s^{-1}\mathcal U\begin{pmatrix}
        0\\0\\\vdots\\\lambda_n 
    \end{pmatrix}=\Phi_s^{-1}\begin{pmatrix}
        \lambda_n\\\vdots\\0 \\0
    \end{pmatrix} = \begin{pmatrix}
        0\\\vdots\\0 \\\lambda_n/\lambda_1
    \end{pmatrix},}
\end{equation}
therefore $\mathbf r$ is mapped to a vector with norm $\lambda_n/\lambda_1>1$.
So, for any non-isotropic (therefore non-commutative) unital CP map $\Phi_s$, there always exists a unitary transformation mapped to a non-positive $\Phi_{t,s}^H$.

The non-unital case follows analogously, by taking  $\mathbf r$ to be in the direction of the maximal shrinking of the Bloch ball, and $\mathcal U$ mapping the direction of maximal shrinking to the direction of minimal shrinking.

\section{Dynamics non-Markovian in both pictures}
\label{app:nM_both}
Assume that $\Phi_{t,s}^S$ is a uniform expansion of the Bloch sphere and consider an arbitrary $\Phi_s$, compatible with $\Phi_t$ being CPTP.
Their Bloch representation can always be written as
\begin{equation}
    \Phi_{s} = \begin{pmatrix}
        1 &0\\\mathbf v^\top  & M
    \end{pmatrix},\qquad\Phi_{t,s}^S = \begin{pmatrix}
        0&0\\0&\lambda I
    \end{pmatrix},\qquad \lambda>1,
\end{equation}
so that
\begin{equation}
    \Phi_{t,s}^H = \begin{pmatrix}
        1&0\\(\lambda-1)M^{-1}\mathbf v^\top & \lambda I
    \end{pmatrix}
\end{equation}
which is also not positive, but (for non-unital dynamics) different from $\Phi_{t,s}^H$.

Similarly, if $ \Phi_{t,s}^H$ is a uniform expansion of the Bloch sphere, then 
\begin{equation}
    \Phi_{t,s}^S = \begin{pmatrix}
        1&0\\(\lambda-1)\mathbf v^\top & \lambda I
    \end{pmatrix}.
\end{equation}
The construction holds for any finite-dimensional Hilbert space.

\section{Contractivity of the classical capacity}
\label{app:contractivity_cl_capacity}
The mutual information can be written as
\begin{equation}
    \label{eq:mutual_info_app}
    \begin{split}
        I(\mathbf p: T(t)\mathbf p) =& {-}\sum_i\big(T(t)\mathbf p\big)_i\log\big(T(t)\mathbf p\big)_i\\ 
        &+ \sum_{ij}T_{ji}(t) p_i\log T_{ji}(t).
    \end{split}
\end{equation}
Assume that $S(t)$ is HPD, so that $T_{ji}(t) = \sum_k T_{jk}(s) {(T_H)}_{ki}(t,s)$;
the second term of Eq.~\eqref{eq:mutual_info_app} can therefore be written as
\begin{equation}
    \begin{split}
        \sum_{ijk}&T_{jk}(s) {(T_H)}_{ki}(t,s)p_i\left[\log T_{jk}(s) +\underbrace{ \log{(T_H)}_{ki}(t,s)}_{\le0}\right]\\
        &\le \sum_{ijk}T_{jk}(s) {(T_H)}_{ki}(t,s)p_i\log T_{jk}(s) \\
        &=\sum_{jk}T_{jk}(s)\left(T_H(t,s)\mathbf p\right)_k \log T_{jk}(s).
    \end{split}
\end{equation}
Substituting in the definition of the capacity one has
\begin{equation}
    \begin{split}
        C\big(T(t)\big) \le & \max_\mathbf p{-} \sum_i\big(T(s)T_H(t,s)\mathbf p\big)_i\log\big(T(s)T_H(t,s)\mathbf p\big)_i\\
        &+\max_{\mathbf p}\sum_{jk}T_{jk}(s)\left(T_H(t,s)\mathbf p\right)_k \log T_{jk}(s)\\
        =&\max_{\mathbf p^\prime = T_H(t,s)\mathbf p}{-} \sum_i\big(T(s)\mathbf p^\prime\big)_i\log\big(T(s)\mathbf p^\prime\big)_i\\
        &+\max_{\mathbf p^\prime = T_H(t,s)\mathbf p} \sum_{jk}T_{jk}(s)p^\prime_k \log T_{jk}(s)\\
        \le&\max_\mathbf p{-}\sum_i\big(T(s)\mathbf p\big)_i\log\big(T(s)\mathbf p\big)_i\\
        &+\max_{\mathbf p} \sum_{jk}T_{jk}(s)p_k \log T_{jk}(s)\\
        =& C\big(T(s)\big).
    \end{split}
\end{equation}

\bibliography{biblio}

\end{document}